\documentclass[11pt]{article}
\usepackage{paper-ak}

\usepackage{graphicx}
\usepackage{amsfonts}
\usepackage{amsmath}
\usepackage{wrapfig}
\usepackage{subcaption}
\usepackage[square, numbers, sort]{natbib} 

\title{Projective Clustering Product Quantization}
\author{Aditya Krishnan \\
	Pinecone, Johns Hopkins University\\
  \texttt{akrish23@jhu.edu}
	\and 
	Edo Liberty \\
	Pinecone\\
  \texttt{edo@pinecone.io}
	}

\date{\nonumber}

\begin{document}
\pagenumbering{arabic}
\maketitle

\begin{abstract}
This paper suggests the use of projective clustering based product quantization for improving nearest neighbor and max-inner-product vector search (MIPS) algorithms. 
We provide anisotropic and quantized variants of projective clustering which outperform previous clustering methods used for this problem such as ScaNN. 
We show that even with comparable running time complexity, in terms of lookup-multiply-adds, projective clustering produces more quantization centers resulting in more accurate dot-product estimates. 
We provide thorough experimentation to support our claims.
\end{abstract}

\section{Introduction}\label{sec:intro}
In vector similarity search one preprocesses $n$ vectors $x \in \R^d$ such that, given a query $q\in R^d$ computing the top (or bottom) $k$ values of $f(x,q)$ is most efficient. 
The function $f$ can encode similarity between $x$ and $q$ like dot product $\langle x, q \rangle$, cosine similarity $\langle x, q \rangle/\|x\|\|q\|$ or Jaccard Similarity. In these cases high values of $f$ are sought. 
Alternatively, $f$ can encode dissimilarity or distance in which case, low values are searched for. Examples include $f(x,q) = \|x-q\|_p = (\sum_i (x_i -q_i)^p)^{1/p}$, or hamming distance. The reader should be aware that all the above problems are highly related and solutions for one often translate to solutions for another. 
For example, max cosine similarity, max dot product and min Euclidean distance are all equivalent for unit vectors. 

In recent years, vector similarity search has taken center stage in the deployment of large scale machine learning applications. 
These include duplicate detection \cite{wang2006large,wang2010towards,zhou2017recent}, personalization \cite{10.1007/978-3-319-56608-5_54,grbovic2018real}, extreme classification \cite{dean2013fast}, and more. 
The demand for vector search was further fueled by proliferation of pre-trained deep learning models which provide semantically rich vector embeddings for text, images, audio, and more \cite{cremonesi2010performance,weston2010large,wu2018starspace}. 
The scale of these problems often requires searching through billions of vectors or performing thousands of searches a second. 
Performance is therefore paramount which necessitates the use of approximate algorithms.
These algorithms trade some tolerance for inaccuracy in the result set with speed.
Practical considerations also need to consider indexing speed, index size, concurrent read-write support, support for inserts and deletes, and many other factors. 
The exact tradeoffs between these factors are both complex and data dependent. 
Due to the importance and complexity of this task, significant academic and industrial efforts were invested in approximate vector search algorithms \cite{johnson2019billion,chen2018sptag,malkov2018efficient}.

Algorithms fall generally into a handful of categories.
Hashing based algorithms \cite{NIPS2014_310ce61c,shrivastava2015assymetric,NIPS2015_2823f479}, tree-search based algorithms \cite{muja2014scalable,dasgupta2008random} and graph based algorithms \cite{malkov2018efficient,harwood2016fanng} mostly focus on reducing the dependence on $n$, the number of vectors. 
Dimension reduction \cite{charikar2002similarity,vempala2005random,li2019random} and quantization algorithms \cite{wu2017multiscale,guo2016quantization,martinez2018lsq++,guo2020accelerating,ge2013optimized,gong2012iterative,martinez2016revisiting,zhang2014composite,ozan2016competitive}, on the other hand, reduce the dependence on $d$ by compressing the index vectors and computing similarities to these compressed representations quickly.
To achieve the best results for a specific application, usually a combination of the above methods is needed.  
Versatile libraries like Faiss \cite{johnson2019billion} along with managed cloud services like Pinecone \cite{pinecone} make these a lot easier for scientists and engineers to use in production applications.

In this paper we focus solely on quantization techniques for approximate max inner product search (MIPS). 
Simple quantization can be achieved with clustering, specifically, a small representative set of $k$ vectors is chosen and each vector in the data is projected onto that set.
Computing the dot products of the query and centers requires only $O(dk)$ time. 
Then, computing the approximate dot product for each point can be done in $O(1)$ time by a simple lookup. 
The shortcoming of this approach is that the quantization is very coarse;
it allows for only $k$ possible quantization centers.
More advanced product quantization techniques divide the $d$ coordinates to $m$ contiguous sections and project each set of $d/m$ coordinates onto one of $k$ centers independently. 
The advantage of this is that computing the dot product of the query and the centers still requires $O(dk)$ operations but now each point can map to one of $k^m$ centers.
The price of this approach is that computing the approximate dot product requires $m$ lookups and additions instead of a single lookup.
Clearly, increasing the values of $m$ and $k$ improves the quality of the approximation while also increasing search running time.

In order to improve the tradeoff between increasing running time and increasing the number of cluster centers, we propose projective clustering product quantization (PCPQ).
The motivation is to significantly increase the number of cluster centers (and thereby improve vector approximation)  
while, at the same time, having negligible increase in application running time.
The exact definition of projective clustering is given in Section \ref{sec:proj_clustering}. 
Unlike standard clustering, each data point is projected to a point along the direction of the center $c$. Specifically, point $i$ projects to $\alpha_i c_j$ where $c_j$ is one of $k$ ``centers'' in $\R^d$ and $\alpha_i \in \R$ is a scalar.
Since $\alpha_i$ is chosen for each point separately, the solution is strictly more general than clustering. 
Our experimental results below corroborate that this indeed provides a significant increase in approximation quality.
Unfortunately, it also requires storing an extra floating point for each point and section of coordinates. 
This is both inefficient in terms of memory consumption and slow in application time since each lookup-add in PQ is replaced with a lookup-multiply-add in PCPQ.

\begin{table}[!t]
\begin{center}
\begin{tabular}{|l|c|c|c|} \hline
                              & Centers   & Application Time   & Index Size          \\ \hline \hline
\parbox{6cm}{Clustering}                      & $k$     & $kd+n$      & $n\log_2(k)$      \\ 
\parbox{6cm}{Product Quantization (PQ)}               & $k^m$     & $kd+2nm$   & $nm\log_2(k)$   \\ 
\parbox{6cm}{Product Quantization (PQ) with $ks$ centers (for comparison)}  & $(ks)^m$  & $skd+2nm$    & $nm\log_2(ks)$    \\ 
Quantized Projective Clustering  (Q-PCPQ)                & $(ks)^m$  & $kd+skm+2nm$   & $nm\log_2(ks)$    \\ \hline
\end{tabular}
\end{center}
\caption{Number of centers, application time to compute inner-products of all points with a query and index size (in bits) for: vanilla clustering, clustering with product quantization, and quantized projective clustering.}
\label{table:clustering_comparison}
\vspace{0.5em}
\hrule
\vspace{-0em}
\end{table}%
To address this, we suggest a quantized version of PCPQ (Q-PCPQ). In Q-PCPQ each $\alpha_i$ is chosen from a set of at most $s$ values. 
As a result, using a lookup table of size $ks$ one can, again, use only $m$ lookup adds. 
Moreover, the number of cluster centers is $(ks)^m$ with Q-PCPQ compared with $k^m$ for PQ.
The increased cost is that the index size is $m\log(ks)$ bits per vector compared to $m\log(k)$ for PQ and the running time requires an additional $skm$ floating point multiplications. 
Note that one could trivially increase the number of centers with PQ to $(ks)^m$ by using $k'=ks$ clusters in each section.
However, with  $k$ typically taking values like $256$, and $s$ taking values like $16$ and $32$, the term $k' d = ksd$ in the query complexity will dominate the query time. Moreover, clustering vectors to $k' = 256 \times 32 \approx 7500$ is not usually feasible during index creation time, especially since the number of points indexed can be as large as $10^9$ in large datasets. We compare the number of centers, application time and index size for these methods in Table \ref{table:clustering_comparison}.

Moreover, we show that the choice of the optimal centers and scaling coefficients for PCPQ and Q-PCPQ can be made to be loss aware. 
Following the logic in \cite{guo2020accelerating}, one can weight the loss of $(\langle x_i,q \rangle - \langle \tilde{x}_i,q \rangle )^2$ higher if $\langle x_i,q \rangle$ is large and therefore $x_i$ is a likely match for $q$.  
Taking that into account, one gets an anisotropic cost function for the optimal cluster centers.  
In our experiments we show that the additional flexibility afforded by the scalars in Q-APCPQ make its recall performance better than ScaNN on standard benchmark datasets for MIPS.

\subsection{Related Works}
\paragraph{Rotation matrix.} In state-of-the-art implementations of PQ \cite{wu2017multiscale,guo2020accelerating,ge2013optimized,gong2012iterative,johnson2019billion}, typically, a rotation matrix $R \in \R^{d \times d}$ is applied to the data before performing clustering and applying PQ. The rotation matrix $R$ is typically learned to make the data more suitable for breaking the $d$ dimensions into sections in PQ. Several methods are used to learn and optimize the matrix $R$, including IPQ \cite{gong2012iterative}, OPQ \cite{ge2013optimized} and MSQ \cite{wu2017multiscale}. Our methods address the method of clustering and performing the product quantization, and are oblivious to pre-processing -- in particular, for a rotation matrix $R \in \R^{d \times d}$, we can approximate the inner-product $\iprod{x}{q}$ with $\sum_{i = 1}^m \iprod{R\alpha^{(i)}_x c^{(i)}_x}{Rq}$ where $\alpha^{(i)}_x, c^{(i)}_x$ denotes the scaling and center from the $i$-th section of Q-PCPQ (or Q-APCPQ) for the data point $x$. A natural future direction to study is whether one can jointly learn the centers, scaling and rotation matrix in order to increase the performance of Q-PCPQ and Q-APCPQ. \vspace{-1em}
\paragraph{Optimization methods.} Several works have focused on proposing novel methods of finding the centers and the assignment of points to centers. For instance, LSQ \cite{martinez2016revisiting} and LSQ++ \cite{martinez2018lsq++} uses a local search based method that outperform the traditional, simpler alternating-minimization based methods. Some methods are proposed to optimize the quality of the centers by introducing cross-training across sections when finding centers, e.g. CompQ \cite{ozan2016competitive} is a stochastic gradient descent based method that achieves this by doing a joint training across sections and CQ \cite{zhang2014composite} is a method that penalizes an error dependent on centers across sections. Our method introduces a novel clustering method, along with a optimization problem for PQ which we optimize using a simple alternating-minimization based approach. We leave the development of more advanced optimization methods for future work.  \vspace{-1em}
\paragraph{Non-PQ Methods.} There are several works that approach MIPS using techniques different than PQ. One popular method is \emph{locality sensitive hashing} LSH, which has seen a flurry of work \cite{NIPS2014_310ce61c,shrivastava2015assymetric,huang2018accurate,neyshabur2015symmetric} for MIPS tasks. In addition, several recent works \cite{sablayrolles2018spreading,Dong2020Learning,erin2015deep,jain2017subic,klein2019end,sablayrolles2017should} have looked at using a neural network to embed the input points and performing the search in the embedded space. Comparison of these methods against PQ methods, such as ours, is a vital research interest to better understand the performance of MIPS systems and we leave this to future work.

\subsection{Preliminaries}
Throughout we denote $X \in \R^{n \times d}$ to be the matrix defined by the data points $x_1, \dots, x_n$ and define $\bar{d} = d/m$ to be the dimension of each of the $m$ sections of coordinates in PQ. We sometimes write $A \subseteq X$ to denote a subset $A$ of the data points in $X$. Additionally, to define the quality of approximation from clustering for MIPS, we denote $\Qc$ to be a distribution over a set of queries in $\R^d$.

\section{Projective Clustering}\label{sec:proj_clustering}
\begin{figure}
  \includegraphics[width=\linewidth]{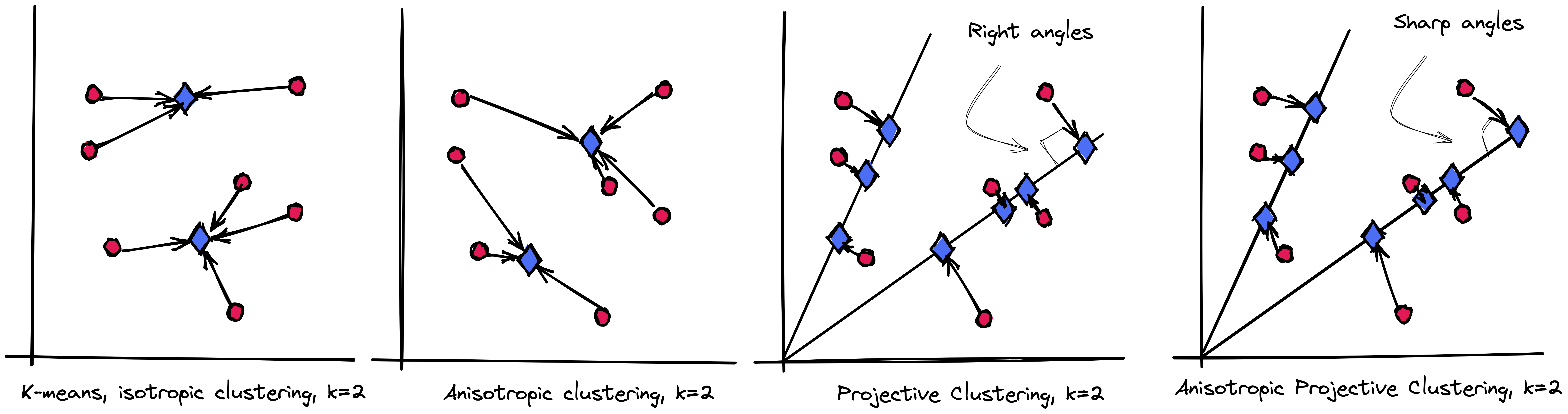}
  \caption{An illustration of $k$-clustering, projective $k$-clustering and their anisotropic counterparts for $k = 2$. While centers in $k$-clustering and anisotropic $k$-clustering are points, for their projective counterparts, they correspond to one-dimensional subspaces. Anisotropic projective clustering has the effect of pushing the centers further away from the origin.}
  \label{fig:clustering}
\vspace{0.5em}
\hrule
\vspace{-1em}
\end{figure}
In projective $k$-clustering, our goal is to find $c_1, \dots, c_k \in \R^{\bar{d}}$ and scalars $\alpha_1, \dots, \alpha_n$ which minimize the following objective: 
\begin{equation}\label{eqn:pq_loss}
\min_{\substack{c_1, \dots, c_k \in \R^{\bar{d}} \\ \alpha_1, \dots, \alpha_n}}  \sum_{i =1}^n \min_{j \in [k]} \Ex_{q \sim \Qc}[ (\iprod{q}{x_i} - \iprod{q}{\alpha_i c_j})^2].
\end{equation} 
When the query distribution $\Qc$ is isotropic, and $\alpha_i$ are constrained to the value $1.0$, the above minimization problem reduces to the ubiquitous \emph{$k$-means clustering} problem. When $\alpha_i$ are allowed to take on any real value, minimizing the loss function from \eqref{eqn:pq_loss} reduces to the popular \emph{$k$-projective clustering} problem; see Appendix \ref{appx:proj_clustering_reduction} for proofs.

\begin{definition}\label{def:projective_clustering}
The \emph{$k$-projective clustering problem} for points $X \in \R^{n \times \bar{d}}$ is a set of $k$ points $c_1, \dots, c_k \in \R^{\bar{d}}$ such that the following is minimized 
\begin{equation}\label{eqn:proj_clustering_loss}
\min_{c_1, \dots, c_k \in \R^{\bar{d}}} \sum_{i = 1}^n \min_{j \in [k]} \left\|x_i - \frac{\iprod{x_i}{c_j}}{\|c_j\|_2^2} \cdot c_j \right \|_2^2.
\end{equation}
In words, the cost incurred by a point $x$ is the cost of projecting it onto the direction (from the choice of $k$ directions given by $c_1, \dots, c_k$) that minimizes the projection cost. 
\end{definition}

For a set of centers $c_1, \dots, c_k \in \R^{\bar{d}}$, and for each $j \in [k]$, let $X_j \subseteq X$ denote the set of all points in $X$ that map to center $c_j$. For each $X_j$, denote $u_j, v_j \in \R^{\bar{d}}$ to be the top left and right singular vectors of $X_j$ respectively and let $\sigma_j \in \R^{\geq 0}$ be the top singular value. 

By the Eckhart-Young-Mirsky Theorem \cite{eckart1936approximation}, we must have that $c_j = v_j$ and $X_jc_j = \sigma_j \cdot u_j$ for all $j \in [x]$ and hence, we can simplify the minimization problem from \eqref{eqn:proj_clustering_loss} as follows: 
\begin{align*}
\min_{\{c_j\}_{j=1}^k} \sum_{i = 1}^n \min_{j \in [k]} \left\|x_i - \frac{\iprod{x_i}{c_j}}{\|c_j\|_2^2} \cdot c_j \right \|_2^2 &= \min_{\{c_j\}_{j=1}^k} \sum_{j = 1}^k \sum_{x \in X_j} \left\|x - \frac{\iprod{x}{c_j}}{\|c_j\|_2^2} \cdot c_j \right \|_2^2  \\
&= \min_{\{X_j\}_{j=1}^k} \sum_{j = 1}^k \sum_{x \in X_j} \left\|x - \sigma_j u_{jx} \cdot v_j \right \|_2^2 \\
&= \min_{\{X_j\}_{j=1}^k} \sum_{j = 1}^k \left\|X_j - \sigma_j u_{j} v_j^\top \right \|_F^2 \\
&= \min_{\{X_j\}_{j=1}^k} \sum_{j = 1}^k \left \|X_j\right \|_F^2 - \sigma_j^2 \\
&= \left\|X\right\|_{F}^2 - \max_{\{X_j\}_{j=1}^k} \sum_{j = 1}^k \sigma_j^2
\end{align*}
where the minimization over $\{X_j\}_{j =1}^k$ is a minimization over all partitions of the input $X$ into $k$ parts. It follows then that once the partition $X_1, \dots, X_k$ is computed, the centers are given by $c_j = v_j$ for $j \in [k]$ and the optimal scalar $\alpha_i$ for the $i$-th point is the projection $\alpha_i = \sigma_ju_{ji}$. 

The running time of applying the projective clustering approximation is asymptotically identical to that of $k$-means and produces significantly more accurate results. This is unsurprising since projective clustering is more expressive than $k$-means. Moreover, the optimal solution for approximating a collection of vectors by a single vector is well known: it is the top singular vector of the collection of vectors rather than their geometric center, which is the optimal solution for $k$-means. 
Building on these encouraging results, we proceed to investigate the quantized version of projective clustering.

\subsection{Quantized Projective $k$-Clustering}\label{sec:quant_proj}
Recall that our goal is to quantize the projections $\alpha_1, \dots, \alpha_n$ of the points $x_1, \dots, x_n$ to $s \in \N$ values. Given the centers $c_1, \dots, c_k$, we can write the quantized version of the minimization problem from above as follows.
\begin{definition}
In the \emph{quantized projective $k$-clustering} problem, a partition $\{X_j\}_{j \in [k]}$ of the data $X$ (corresponding to clusters) is given and the goal is to find a set of vectors $\{\bar{u}_j \in \R^{\bar{d}}\}_{j \in [k]}$ which contain at most $s$ distinct entries while minimizing the following objective: 
\begin{equation}\label{eqn:qpcpq_loss}
\min_{\{\bar{u}_j\}_{j\in[k]}} \sum_{j =1}^k \|X_j - \bar{u}_jv_j^\top\|_F^2
\end{equation}
where $v_j \in \R^{\bar{d}}$ is the top right singular vector of $X_j$. 
\end{definition} 
It can be shown that in fact, the objective function from \eqref{eqn:qpcpq_loss} is upper bounded by the cost of the optimal clustering plus the cost of doing a one-dimensional $k$-means clustering of the projections $\alpha_1, \dots, \alpha_n$ with $k = s$. We state this specifically in the following fact and give a proof in Appendix \ref{appx:qpcpq}.  
\begin{fact}\label{fact:qpcpq_loss_reduction}
The loss in \eqref{eqn:qpcpq_loss} is upper-bounded by: 
$$ \min_{\{\bar{u}_j\}_{j\in[k]}} \sum_{j =1}^k \|X_j - \sigma_ju_jv_j^\top\|_F^2 + \|\sigma_ju_j - \bar{u}_j\|_2^2 $$
where $u_j$ and $\sigma_j$ are the left singular vector and top singular value of $X_j$ respectively. 
\end{fact}
It is now easy to see that minimizing the loss $\sum_{j = 1}^k\|\sigma_ju_j - \bar{u}_j\|_2^2$ is simply a $k$-means problem in one dimension with $k = s$, hence this can be easily computed after finding the centers and the projections for each section.

\section{Anisotropic Projective Clustering}
\begin{figure}[!t]
\begin{center}
  \includegraphics[width=0.8\linewidth]{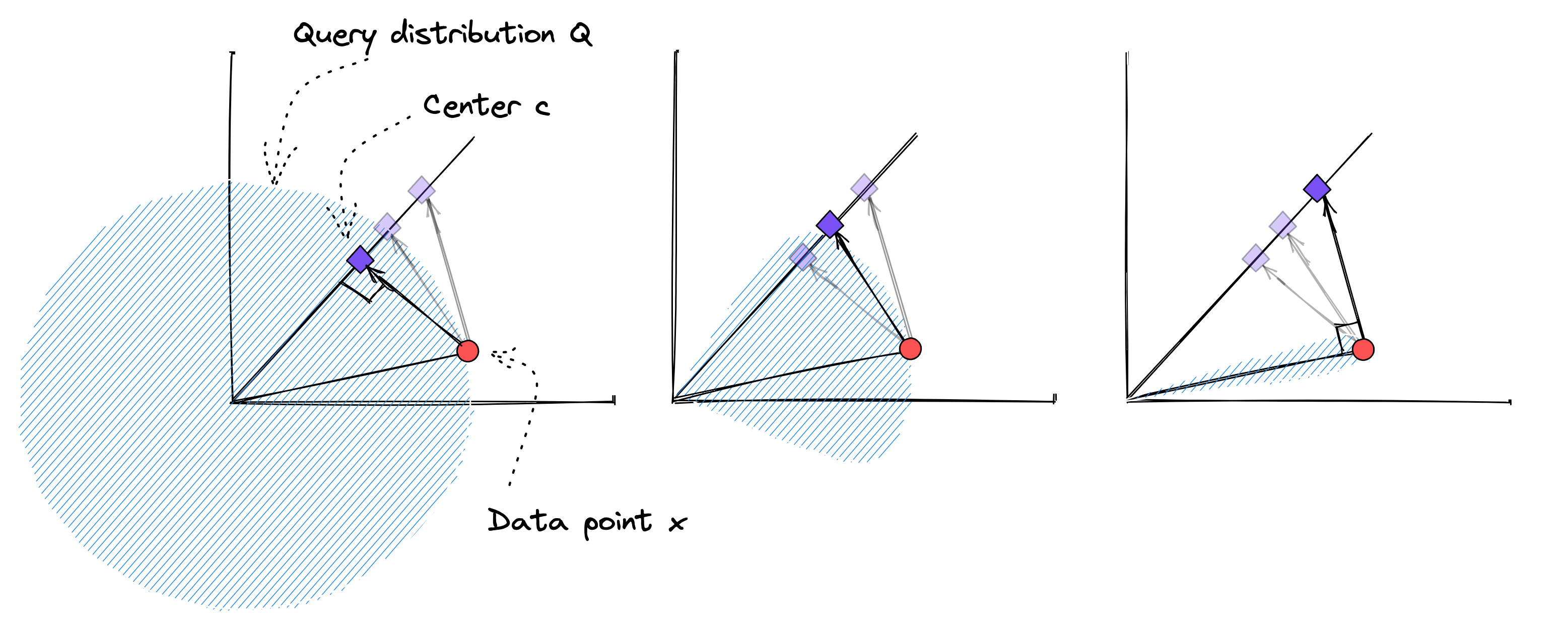}
  \caption{The optimal scaling $\alpha$ for anisotropic projective clustering with differing query distributions. The weight function $\id_t : \R \to \{0, 1\}$ essentially restricts the queries in the distribution $\Qc$ on which the loss $\langle q, x-c \rangle^2$ is measured by outputting $0$ on all the queries $q$ for which $\langle q ,x \rangle < t$. When the distribution is isotropic, $x$ is mapped to its projection. When $q = x$, the point $x$ is mapped to a center $c$ such that the inner-product $\langle c, q\rangle = \langle x, q \rangle = \|x\|_2^2$.
  }
  \label{fig:anisotropic_clustering}
\vspace{0.5em}
\hrule
\vspace{-0em}
\end{center}
\end{figure}
The objective function in \eqref{eqn:pq_loss} takes the expectation over the entire query distribution $\Qc$, however, as observed by \citet{guo2020accelerating}, not all pairs of points and queries are equally important since in the MIPS setting we are interested in preserving the inner-product $\iprod{x}{q}$ only if it is large (and hence a candidate for a nearest neighbor). This observation motivates the work of \citet{guo2020accelerating} that introduces a weight function in the loss given by \eqref{eqn:pq_loss}. Specifically, given a threshold $t > 0$, let $\id_t(x)$ be a weight function that is $1$ when $x \geq t$ and $0$ when $x < t$. Then, they consider the following \emph{score-aware loss} function for clustering: 
\begin{equation}\label{eqn:anisotropic_loss}
\min_{\substack{c_1, \dots, c_k \in \R^{\bar{d}} \\ \alpha_1, \dots, \alpha_n}} \Ex_{q \sim \Qc}[ \sum_{i =1}^n \min_{j \in [k]} \id_t(\iprod{x_i}{q}) \cdot (\iprod{q}{x_i} - \iprod{q}{ \alpha_ic_j})^2]
\end{equation}
where $\alpha_i$ is fixed to be $1.0$ and where $t > 0$ is a tuneable input-specific parameter\footnote{When the input is unit norm, it is suggested to set $t = 0.2$ \cite[Section 3.2]{guo2016quantization}.}. The threshold function $\id_t$ defines a cone of vectors (of unit length) around the point $x_i$. It  essentially ``selects'' the queries from the support of $\Qc$ for which the inner-product between $x_i$ and the queries should be preserved. Figure \ref{fig:clustering} depicts this intuition, showing that anisotropic (projective) clustering finds centers that have larger norm than in plain vanilla clustering so as to preserve the inner-products with the relevant queries, defined by the threshold function.

As shown in \citet{guo2020accelerating}, the above loss (per point) can be written as a linear combination of the error introduced by the component of the residual $x_i - c_j$ that is parallel to $x_i$ and the component that is orthogonal to $x_i$. Formally, define $r_\parallel(x, c) \eqdef x - \frac{\iprod{x}{c}}{\|x\|_2^2}x$ and $r_\bot(x, c) \eqdef  \frac{\iprod{x}{c}}{\|x\|_2^2}x - c$, then, it was shown in \citet[Theorem 3.2]{guo2020accelerating} that the loss in \eqref{eqn:anisotropic_loss} is equivalent to minimizing the loss 
$$\min_{\substack{c_1, \dots, c_k \in \R^{\bar{d}}}} \sum_{i = 1}^n \min_{j \in [k]} h_\parallel(\|x_i\|_2) \cdot \| r_\parallel(x, c_j) \|_2^2 + h_\bot(\|x_i\|_2) \cdot \|r_\bot(x, c_j)\|_2^2$$
where $h_\parallel(y) = (\bar{d}-1)\int_0^{t/y} \sin^{\bar{d}-2}\theta - \sin^{\bar{d}} \theta d\theta$ and $h_\bot(y) = \int_0^{t/y} \sin^{\bar{d}}\theta d\theta$. 

Intuitively, the functions $h_\parallel$ and $h_\bot$ represent how the error from the parallel component of the residual and the orthogonal component of the residual should be weighted depending on the weight-function $\id_t$ in the loss \eqref{eqn:anisotropic_loss}. It can be shown that for any $t \geq 0$ we have that $h_\parallel \geq h_\bot$, i.e. the parallel component of the error is weighted higher. We defer further discussion of $h_\parallel$ and $h_\bot$ to Appendix \ref{appx:hvalues_discussion}.

A natural way to extend the above loss, as in projective $k$-clustering, is to remove the restriction on the value of $\alpha_1, \dots, \alpha_n$ and allow it to take on any real value. We term this the \emph{anisotropic projective $k$-clustering problem} which we define below. 
\begin{definition}\label{def:projective_anisotropic_clustering}
The \emph{anisotropic projective $k$-clustering problem} for a set of points $X \in \R^{n \times \bar{d}}$ aims to find a set of $k$ points $c_1, \dots, c_k \in \R^{\bar{d}}$ and $n$ scalars $\alpha_1, \dots, \alpha_n \in \R$ such that the following is minimized
\begin{equation}\label{eqn:proj-anisotropic-loss}
\min_{\substack{c_1, \dots, c_k \in \R^{\bar{d}} \\ \alpha_1,\ldots,\alpha_n \in \R}}\sum_{i = 1}^n \min_{j \in [k]} \left(h_\parallel(\|x_i\|_2) \cdot \|r_\parallel(x_i, \alpha_ic_j)\|_2^2 + h_\bot(\|x_i\|_2) \cdot \|r_\bot(x_i, \alpha_ic_j)\|_2^2 \right).
\end{equation}
\end{definition}
We show from a simple calculation that in fact, for a fixed $c_j$ and $x_i$, one can compute the optimal $\alpha_i \in \R$ that minimizes the loss in $\eqref{eqn:proj-anisotropic-loss}$. Intuitively, the optimal scaling pushes the center away from the projection in the direction of $c_j$ depending on which queries have weight $1$ in the loss function \eqref{eqn:proj-anisotropic-loss}; Figure \ref{fig:anisotropic_clustering} depicts this intuition.  
\begin{fact}\label{fact:opt_alpha_aniso}
For vectors $x, c \in \R^{\bar{d}}$, we have 
$$\argmin_{\alpha \in \R} \ h_\parallel \cdot \left\|x - \frac{\iprod{x}{\alpha c}}{\|x\|_2^2}x \right\|_2^2 + h_\bot \cdot \left\|\alpha c - \frac{\iprod{x}{\alpha c}}{\|x\|_2^2}x \right\|_2^2 = \frac{h_\parallel \cdot \iprod{x}{c}}{\frac{(h_\parallel - h_\bot)\iprod{x}{c}^2}{\|x\|_2^2} + h_\bot \cdot \|c\|_2^2}$$

where $h_\parallel$ and $h_\bot$ are abbreviations for $h_\parallel(\|x\|_2)$ and $h_\bot(\|x\|_2)$ respectively. 
\end{fact}

\subsection{Quantized Anisotropic Projective Clustering}\label{sec:quant_aniso}
The quantized version of anisotropic projective $k$-clustering can be written by modifying the minimization problem from \eqref{eqn:proj-anisotropic-loss} like it was done for projective $k$-clustering. Specifically, given the centers $c_1, \dots, c_k$, the goal is to minimize the following loss: 
\begin{equation}
\min_{\substack{\lambda_1,\ldots,\lambda_s}}\sum_{i = 1}^n \min_{\substack{j \in [k], \ l \in [s]}} \left(h_{\parallel}(\|x_i\|_2) \cdot \|r_\parallel(x_i, \lambda_lc_j)\|_2^2 + h_\bot(\|x_i\|_2) \cdot \|r_\bot(x_i, \lambda_lc_j)\|_2^2 \right). \label{eqn:quantized_aniso_proj_loss}
\end{equation}
Using the definition of $r_\parallel$ and $r_\bot$, the term in the summation can be written in the form $w_i\lambda_l^2 + a_i\lambda_l + b_i$ where $w_i, a_i, b_i \in \R$ are constants\footnote{$w_i = (h_\parallel  +h_\bot)\iprod{x_i}{c_j}(\iprod{x_i}{c_j}-2) + h_\bot\|c_j\|_2^2$, $a_i = -2\iprod{x_i}{c_j}$ and $b_i = h_\parallel \|x_i\|_2^2$.} that are independent of $\lambda_1, \dots, \lambda_s$ and only dependent on the point $x_i$ and the center it maps to. Hence, the above minimization problem can be simplified as follows: 
$$\min_{\lambda_1, \dots, \lambda_s} \sum_{i = 1}^n \min_{l \in [s]} w_i\lambda_l^2 + a_i\lambda_l + b_i.$$ 
Notice that this is a quadratic loss for $s$ variables in one-dimension for which we can find a fixed point efficiently by doing an alternating minimization (with random initialization) since the minimization problem can be solved exactly for a single scalar.

\section{Efficient Algorithms and Application to PQ}
In this section we show how to implement Q-PCPQ and Q-APCPQ in two stages: first we give efficient algorithms in Section \ref{sec:aniso_proj_clustering_alg} and \ref{sec:aniso_proj_clustering_alg} to compute the centers and scalars for projective $k$-clustering (PCPQ) and anisotropic projective $k$-clustering (APCPQ) respectively, and then, in Section \ref{sec:pq_alg}, we state how to create the index for Q-PCPQ and Q-APCPQ by using the methods from Section \ref{sec:quant_proj} and \ref{sec:quant_aniso} to quantize the scalars we computed and store the mapping of the points to their respective centers and quantized scalars. In the same section, we show how the index for Q-PCPQ and Q-APCPQ is used to compute the inner-product between a query vector and every point. 

Our main approach to compute the centers and scalars is to use an alternating-minimization algorithm to minimize the loss function \eqref{eqn:proj_clustering_loss} for PCPQ and \eqref{eqn:proj-anisotropic-loss} for APCPQ. Alternating-minimization has been a long-standing approach on finding the clustering in the literature for product quantization \cite{gong2012iterative,ge2013optimized,wu2017multiscale,guo2020accelerating}.

\subsection{Projective $k$-clustering}\label{sec:proj_clustering_alg}
The projective $k$-clustering problem is a NP-Hard problem for $k \geq 2$ even in the 2D-plane, i.e. $\bar{d} = 2$. Nevertheless, several lines of work have proposed fast approximation algorithms -- including algorithms inspired by Lloyd's algorithm for k-means that iteratively find $k$ one-dimensional subspaces \cite{agarwal2004k}, a monte-carlo sampling method \cite{procopiuc2002monte} and coreset based algorithms \cite{feldman2020turning,statman2020faster} that sample (and reweight) a small subset of the input points in such a way that the cost of clustering the subset of points is approximately the same as clustering the entire input, allowing the use of any black-box algorithm (that is potentially computationally intractable on large datasets) on the much smaller set of points.

The algorithm we propose is similar to Lloyd's algorithm and iterates between an \emph{assignment step} and a \emph{center finding step} -- in each round the algorithm i) maps the points to the closest center by projecting every point to each of the directions (given by the centers) and picks the one that minimizes the squared distance of the point to its projection and ii) for each center, the algorithm considers all the points that mapped to the center in step i) and computes the top-right singular vector of those points to be the new center. It is easy to see from the Eckart-Young-Mirsky Theorem \cite{eckart1936approximation} that the top-right singular vector in fact minimizes the loss from \eqref{eqn:proj_clustering_loss} when $k = 1$.

Formally, we consider the following iterative procedure: 
\begin{enumerate}[1.]
	\item{(Initialization)} Sample $k$ random points $c_1, \dots, c_k$ from the data $X$. 
	\item{(Assignment)} For each point $x \in X$: let $j^* = \argmin_{j \in [k]} \|x - \frac{\iprod{x}{c_j}}{\|c_j\|_2^2}c_j\|_2$, set $\alpha_i = \frac{\iprod{x}{c_{j^*}}}{\|c_{j^*}\|_2^2}$ and assign $c_{j^*}$ to be the center of $x$.
	\item{(Center Finding)} For each $j \in [k]$: denote $X_j \subset X$ be the subset of the points that were assigned to center $j$ in Step 2, compute the top-right singular vector $v$ of $X_j$, i.e. $v = \argmin_{\|y\|=1}\|X_j - X_jyy^\top\|_F^2$, and set $c_j = v$.  
	\item{(Termination)} Repeat Step 2 and 3 until convergence or a maximum number of iterations has been reached.
\end{enumerate}

\paragraph{Initialization.}
In practice, as in our experiments, we can hope to find better solutions (centers) by initializing the centers using an off-the-shelf algorithm for another clustering problem, like $k$-means++. Recently, the work of \cite{statman2020faster} proposed a different initialization that provably produces a $O(\log(k))$-approximation for the projective $k$-clustering problem. The algorithm first normalizes all the points in $X$ so they have unit norm and then samples $k$ points with a slight variant of the k-means++ initialization algorithm -- specifically, the algorithm maintains a set $S$ and, for $k$ iterations, samples a point $x \in X \backslash S$ to add to $S$ with probability proportional to $\min(\min_{s \in S}\|x - s\|_2^2, \min_{s \in S}\|x + s\|_2^2)$. Empirically, on the datasets we use, we found that initializing the centers with $k$-means++ produces centers with similar loss as this algorithm while being able to take advantage of optimized implementations of $k$-means++ leading to significant gains in indexing time. 

\subsection{Anisotropic Projective $k$-clustering}\label{sec:aniso_proj_clustering_alg}
We propose an alternating-minimization framework for anisotropic projective $k$-clustering, similar to the algorithm from Section \ref{sec:proj_clustering_alg}. The difference lies in the assignment stage, where we need to compute the optimal $\alpha_i$ for each data point $x_i \in X$ with respect to the anisotropic projective clustering loss from Definition \ref{def:projective_anisotropic_clustering} and in the center finding stage where, for each subset of points in $X$ that are mapped to the same center, we need to compute the center that minimizes the aforementioned loss.

We showed in Fact \ref{fact:opt_alpha_aniso} that computing the optimal $\alpha_i$ for each $x_i$ can be done when the centers are fixed. It is left then to show how to compute the centers in the center finding stage ($k=1$). \citet{guo2020accelerating} show that when $\alpha_i = 1$, there is a closed form solution to finding the center that minimizes the loss from Definition \ref{def:projective_anisotropic_clustering}. A simple adaptation of their proof gives us a closed form solution to the aforementioned loss when $\alpha_i$ is any arbitrary fixed scalar.
\begin{theorem}[Theorem 4.2, \cite{guo2020accelerating}]\label{thm:opt_center_anisotropic_clustering}
For $n$ points $X \in \R^{n \times \bar{d}}$, we have that  
$$c^* \eqdef \brackets{\sum_{i=1}^n\frac{\alpha_i^2(h_{\parallel, i} - h_{\bot, i})}{\|x_i\|_2^2}x_ix_i^\top + I \cdot \alpha_i^2h_{\bot, i}}^{-1}\sum_{i = 1}^n \alpha_i h_{\parallel, i}x_i$$
is the solution to the minimization problem
$\argmin_{c \in \R^{\bar{d}}} \sum_{i = 1}^n h_{\parallel, i} \|r_\parallel(x_i, \alpha_i c)\|_2^2 + h_{\bot, i} \|r_\bot(x_i, \alpha_i c)\|_2^2$ 
and where $h_{\parallel, i}$ and $h_{\bot, i}$ are abbreviations for $h_\parallel(\|x_i\|_2)$ and $h_\bot(\|x_i\|_2)$ respectively.
\end{theorem}

We propose the following iterative procedure for the anisotropic projective $k$-clustering problem. 
\begin{enumerate}[1.]
	\item{(Initialization)} Sample $k$ random points $c_1, \dots, c_k$ from the data $X$. 
	\item{(Assignment)} For each point $x_i \in X$: compute $\beta_{i}[j] \eqdef {h_{\parallel, i}  \iprod{x_i}{c_j}} ({\frac{(h_{\parallel, i} - h_{\bot, i})\iprod{x_i}{c_j}^2}{\|x_i\|_2^2} + h_{\bot,i} \cdot \|c_j\|_2^2})^{-1}$ for each $j \in [k]$. Compute $j^* = \argmin_{j \in [k]} \|x_i - \beta_{i}[j]c_j\|_2$, then, set $\alpha_i = \beta_{i}[j^*]$ and $c_{j^{*}}$ to be the center for $x_i$.
	\item{(Center Finding)} For each $j \in [k]$: denote $X_j \subset X$ be the subset of the points that were assigned to center $j$ in Step 2, set $c_j$ to be
	$$\brackets{\sum_{\substack{i \in [n] \text{ s.t.}\\ x_i \in X_j}}\frac{\alpha_i^2(h_{\parallel, i} - h_{\bot,i})}{\|x_i\|_2^2}xx^\top + I \cdot \alpha_i^2h_{\bot, i}}^{-1}\sum_{\substack{i \in [n] \text{ s.t.}\\ x_i \in X_j}} \alpha_i h_{\parallel, i} \cdot x_i$$  
	\item{(Termination)} Repeat Step 2 and 3 until convergence or a maximum number of iterations has reached.
\end{enumerate}
Notice that unlike in our algorithm for projective $k$-clustering, the center finding stage in the above algorithm does not find the optimal scaling and center for each subset (cluster) of points. Instead we solve the anisotropic $k$-clustering problem on the cluster and then compute the optimal scaling using Fact \ref{fact:opt_alpha_aniso}. Computing the optimal scaling and center jointly, like we did for vanilla projective clustering, is an open problem we leave for future work.

\subsection{Product Quantization and Running Time}\label{sec:pq_alg}
Next, we outline how we use the algorithms from Sections \ref{sec:proj_clustering_alg} and \ref{sec:aniso_proj_clustering_alg} to compute the index for Q-PCPQ and Q-APCPQ respectively and how the inner-product of the query vector with each point is computed during query time.

As discussed in Section \ref{sec:intro}, the columns of the input data $X \in \R^{n \times d}$ are split into $m$ contiguous chunks $X^{(1)}, \dots, X^{(m)} \in \R^{n \times d/m}$ where $X^{(j)}$ contains the coordinates $[jm, \dots, (j+1)m - 1]$. We then compute the centers $C^{(1)}, \dots, C^{(m)} \in \R^{k \times d/m}$ and scalars $\{\alpha^{(j)}_i : j \in [m], i \in [n]\}$ independently for each section $j \in [m]$. The centers and scalars for each section are computed using the method in Section \ref{sec:proj_clustering_alg} for Q-PCPQ and in Section \ref{sec:aniso_proj_clustering_alg} for Q-APCPQ.  

We store $m$ lookup tables $\{\phi_{j} : [n] \to [k] \ | \ j \in [m] \}$, where look up table $\phi_j$ maps every point to its center in $C^{(j)}$.

In order to quantize the scalars, we follow the procedure from Section \ref{sec:quant_proj} for Q-PCPQ and Section \ref{sec:quant_aniso} for Q-APCPQ\footnote{Recent implementations of product quantization for Euclidean distance \cite{wu2017multiscale} and that of anisotropic $k$-clustering (ScaNN) \cite{guo2020accelerating} unit normalize all the data and store the norm $\|x_i\|_2$ as a scalar $\alpha_i$ for each $x_i$ separately. The scalars are then quantized to $s$ values, and during query time, the inner-product computations for each point are scaled by its respective quantized norm.}. Specifically, the scalars $\{\alpha^{(j)}_i \ | \ j \in[m], i \in [n]\}$ are themselves quantized to $s$ values $\{\lambda_1, \dots, \lambda_s\}$ using the aforementioned procedured and $m$ mappings $\{\gamma_j : [n] \to [s] \ | \ j \in [m]\}$ are created such that mapping $\gamma_j$ maps the scalars $\{\alpha^{(j)}_i \}_{i \in [n]}$ from section $j$ to their respective quantized scalar among $\{\lambda_1, \dots, \lambda_s\}$.


\paragraph{Inner-product computation.} For a query $q \in \R^d$, let $[q_1, \dots, q_m]$ denote the $m$ sections of the coordinates of $q$ each containing $d/m$ coordinates. The inner-products between the data $X$ and $q$ is computed in the following phases: 
\begin{enumerate}[i.]
\item for each $j \in [m]$, the $k$-length vector $\eta_j \eqdef C^{(j)}q_j$ is computed, then, 
\item for each $j \in [m]$ and $l \in [s]$, the quantity $\eta_j \cdot \lambda_l$ is computed and stored in a lookup table of size $skm$, and finally, 
\item to compute the inner-product with $x_i$, the sum \begin{equation}\label{eqn:final_ip_sum}\sum_{j = 1}^m \eta_{j, \phi_j(i)} \cdot \lambda_{\gamma_j(i)}\end{equation} is computed by performing $m$ lookups, one for each term.
\end{enumerate}

\paragraph{Running time and space complexity.} Computing the inner-products of the centers in each section with $q$ (i.e. stage i.) takes $kd$ multiply-adds, since each section has $d/m$ coordinates and there are $m$ sections with $k$ centers each. Computing the products with the $s$ scalars $\lambda_1, \dots, \lambda_s$ requires $skm$ multiplication operations, $k$ for each scalar and section. And finally computing the inner-product in step iii. requires $m$ lookups and $m$ additions per point, for a total of $nm$ lookups and additions. In total, the complexity of the search is $kd + skm + 2nm$ operations. 

Storing the mapping $\phi^{(1)}, \dots, \phi^{(m)}$ requires $m\log_2(k)$ bits per point and storing the scalar quantizing requires $\log_2(s)$ bits per point, leading to a total of $n(m\log_2(k) + \log_2(s))$ bits to store the index. In addition, $kd + s$ floating point integers need to be stored, which is a lower order term since $n \gg kd/m$.

\section{Experiments}
\begin{figure}[t!]
\centering
\begin{subfigure}{0.6\linewidth}
  \centering
  \includegraphics[width=\linewidth]{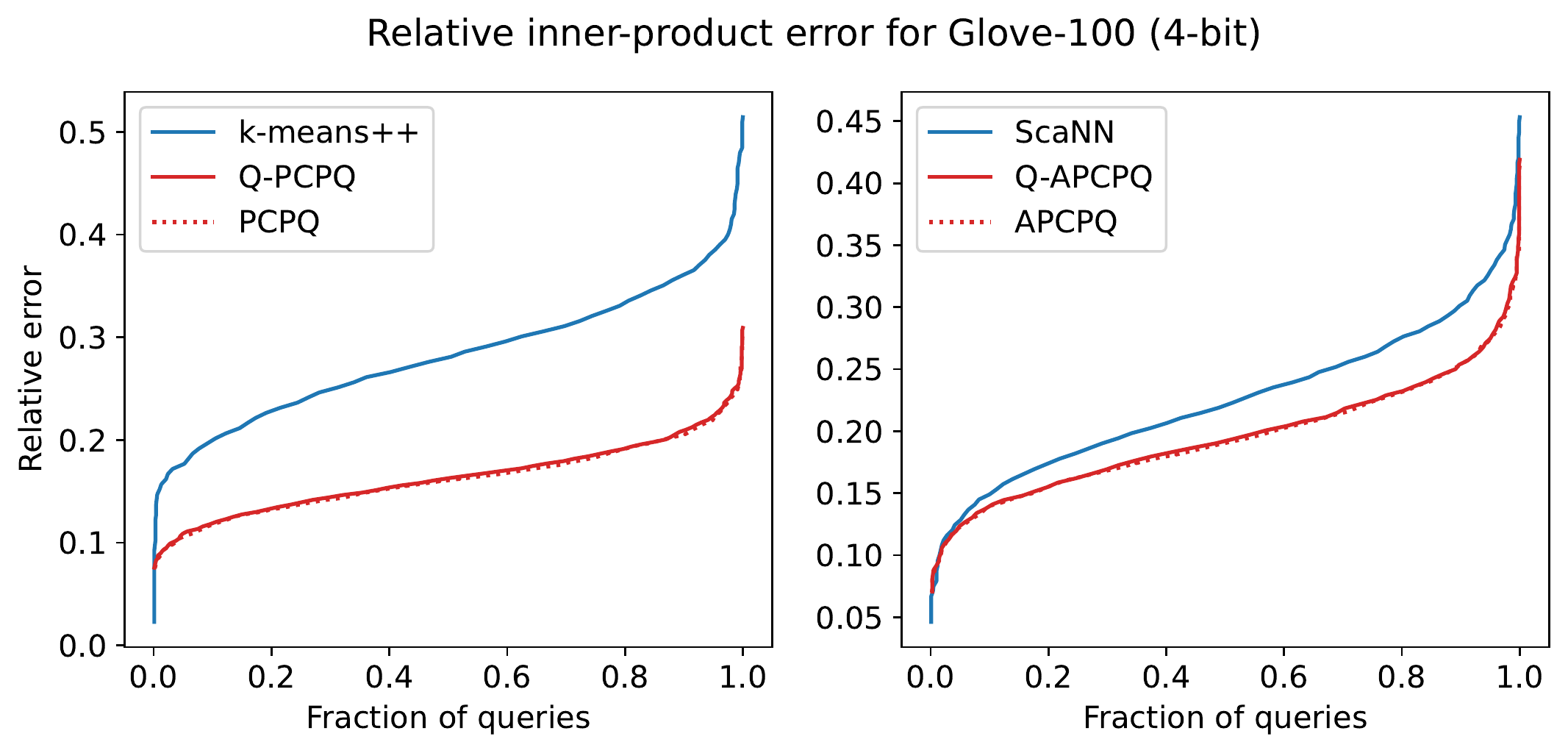}
  \caption{Relative error of each method over 1000 queries, sorted by the magnitude of the error with the query.}
  \label{fig:ip_error_glove_plot}
\end{subfigure}%
\begin{subfigure}{0.4\linewidth}
  \centering
  \begin{small}
  \begin{tabular}{|l |l |l|}
  \hline
  Method (4-bit) & Glove-100 & Last.fm \\
  \hline 
  \hline
$k$-means++    & 0.286     & 0.191   \\
Q-PCPQ         & 0.166     & 0.124   \\
PCPQ           & 0.164     & 0.114   \\
\hline
ScaNN          & 0.229     & 0.178   \\
Q-APCPQ        & 0.199     & 0.153   \\
APCPQ          & 0.197     & 0.155  \\ 
\hline
  \end{tabular}
  \end{small}
  \caption{Average relative error of inner-product over 1000 queries.}
  \label{fig:ip_error_table}
\end{subfigure}
\caption{Comparison of 4-bit $k$-means++, ScaNN, Q-PCPQ and Q-APCPQ (along with the versions without scalar quantization: PCPQ and APCPQ) on the relative error of approximating the top inner-product over 1000 queries on Glove-100 and Last.fm. }
\label{fig:ip_error_glove}
\vspace{0.5em}
\hrule 
\vspace{-0em}
\end{figure}

\begin{figure}[t]
\centering
\begin{subfigure}{0.5\linewidth}
  \centering
  \includegraphics[width=1.\linewidth]{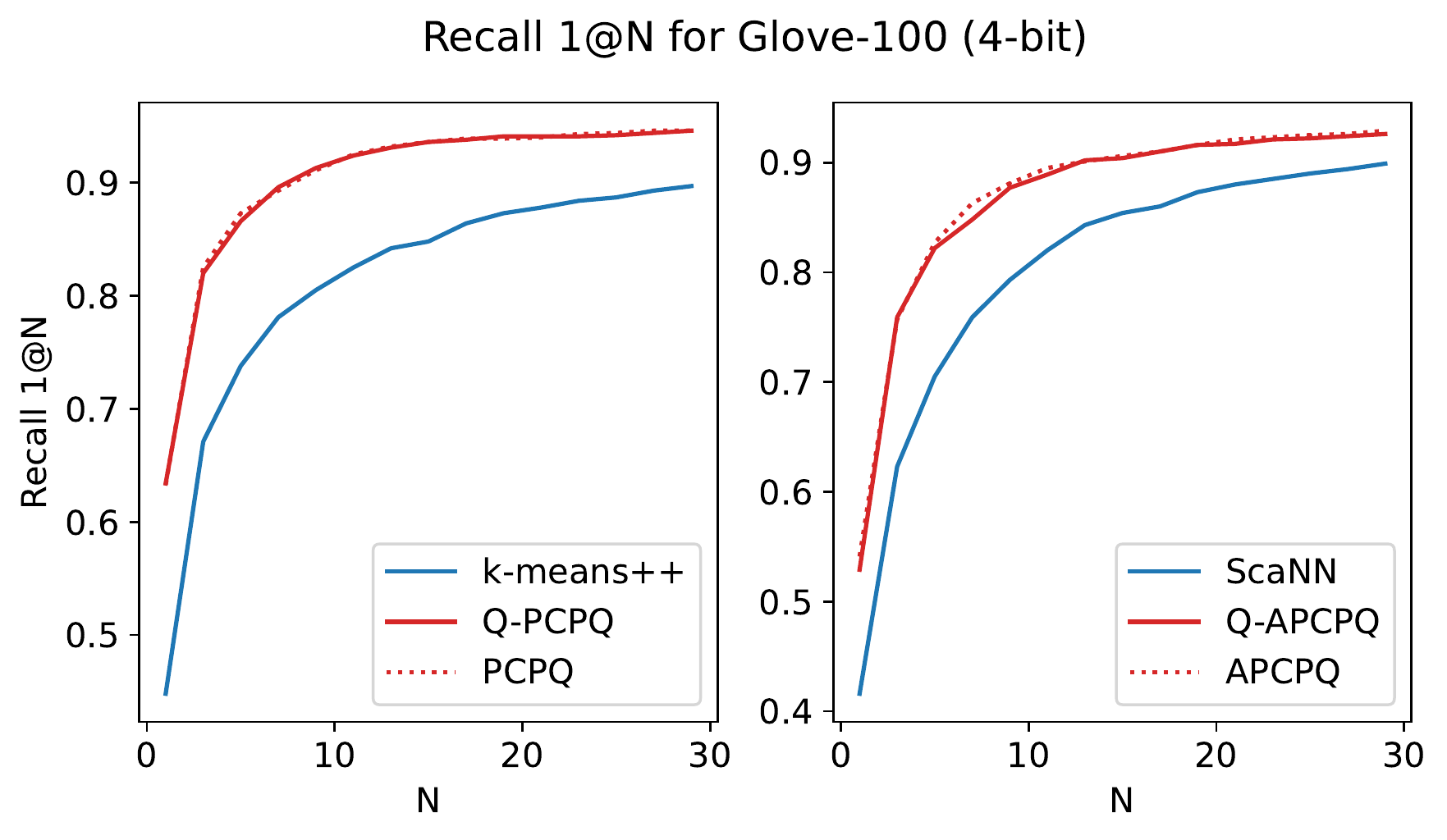}
  \label{fig:recall_glove_plot_4bit}
\end{subfigure}%
\begin{subfigure}{0.5\linewidth}
  \centering
  \includegraphics[width=1.\linewidth]{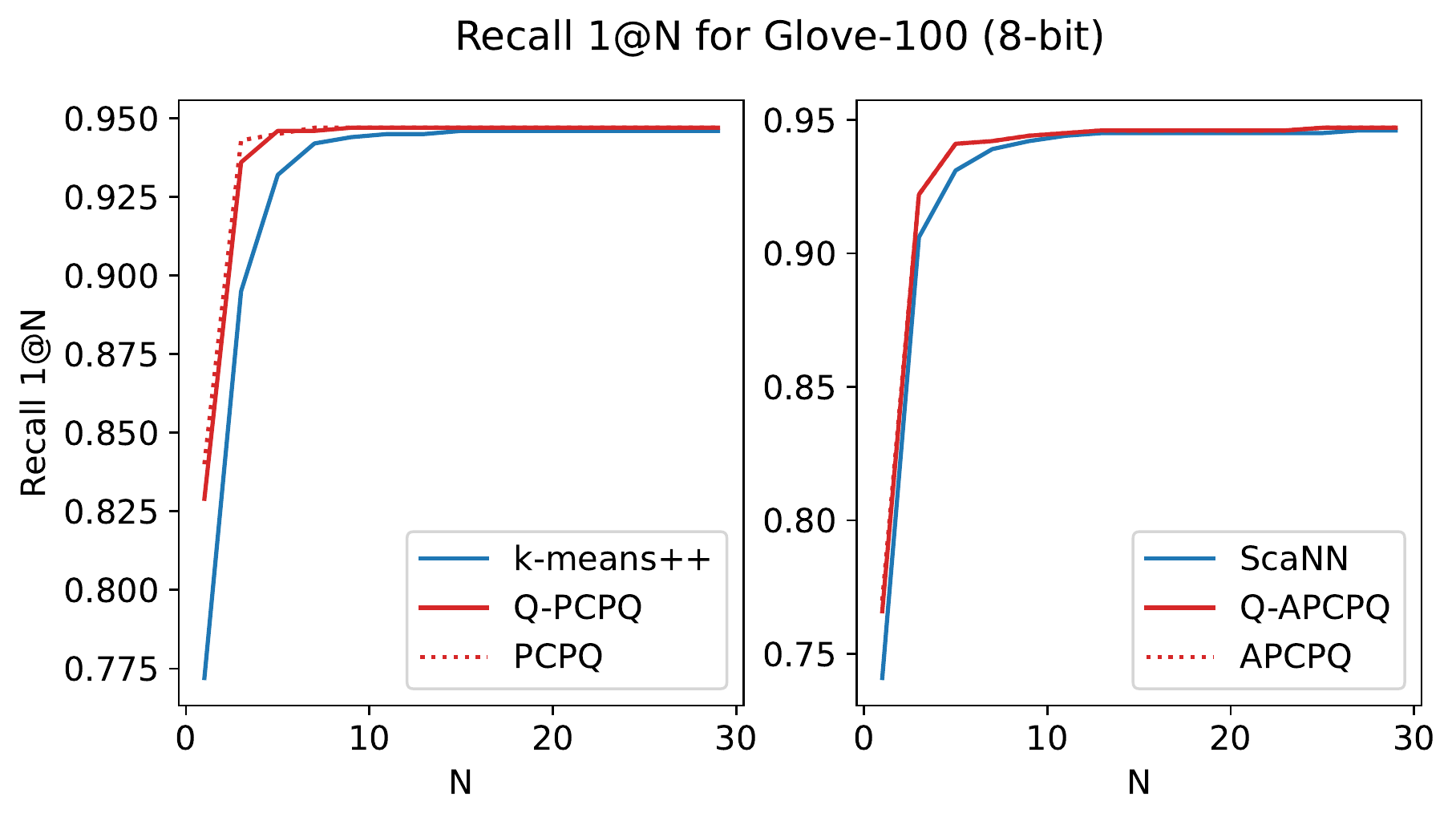}
  \label{fig:recall_glove_plot_8bit}
\end{subfigure}%
\newline
\begin{subfigure}{0.5\linewidth}
  \centering
  \includegraphics[width=1.\linewidth]{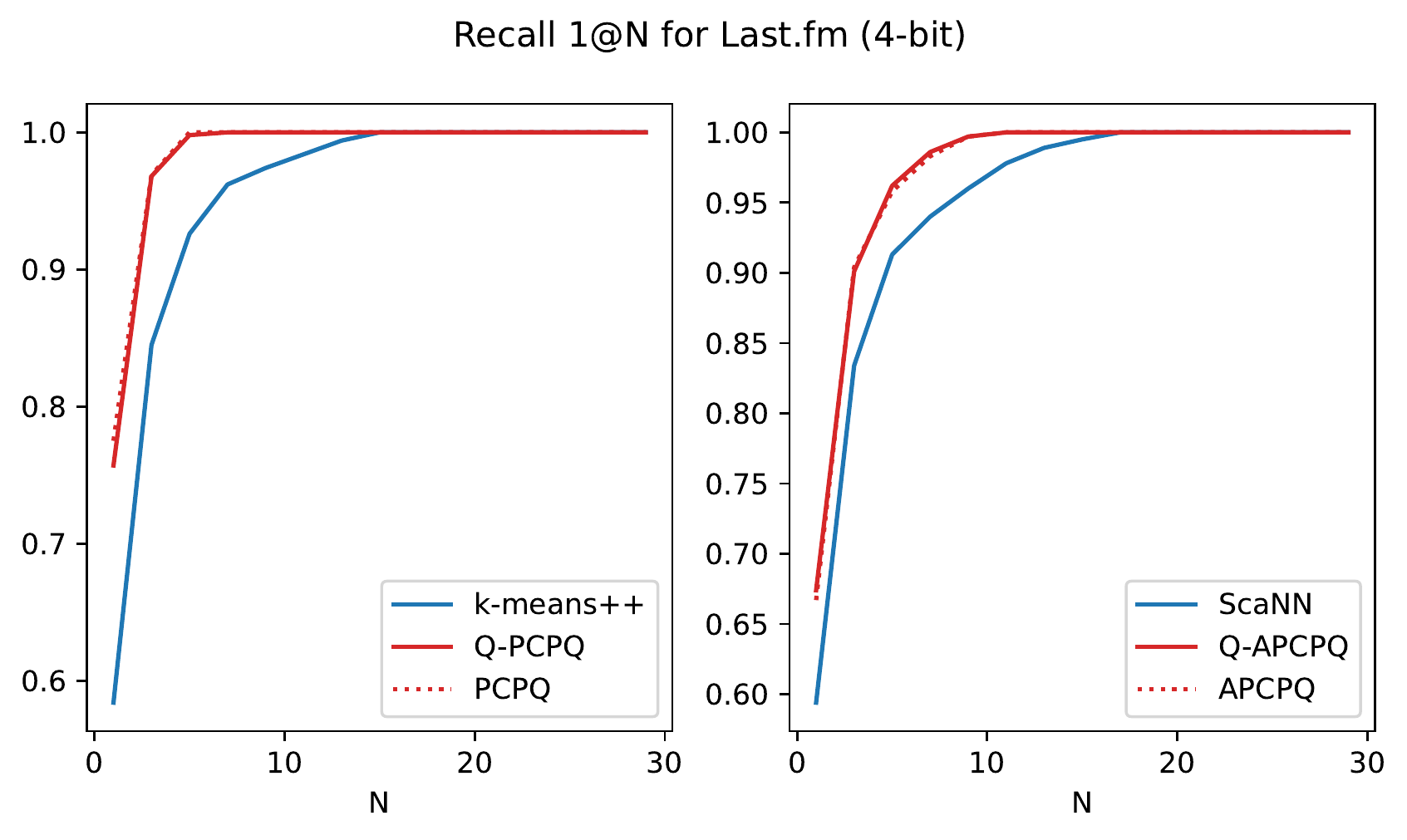}
  \label{fig:recall_lastfm_plot_4bit}
\end{subfigure}%
\begin{subfigure}{0.5\linewidth}
  \centering
  \includegraphics[width=1.\linewidth]{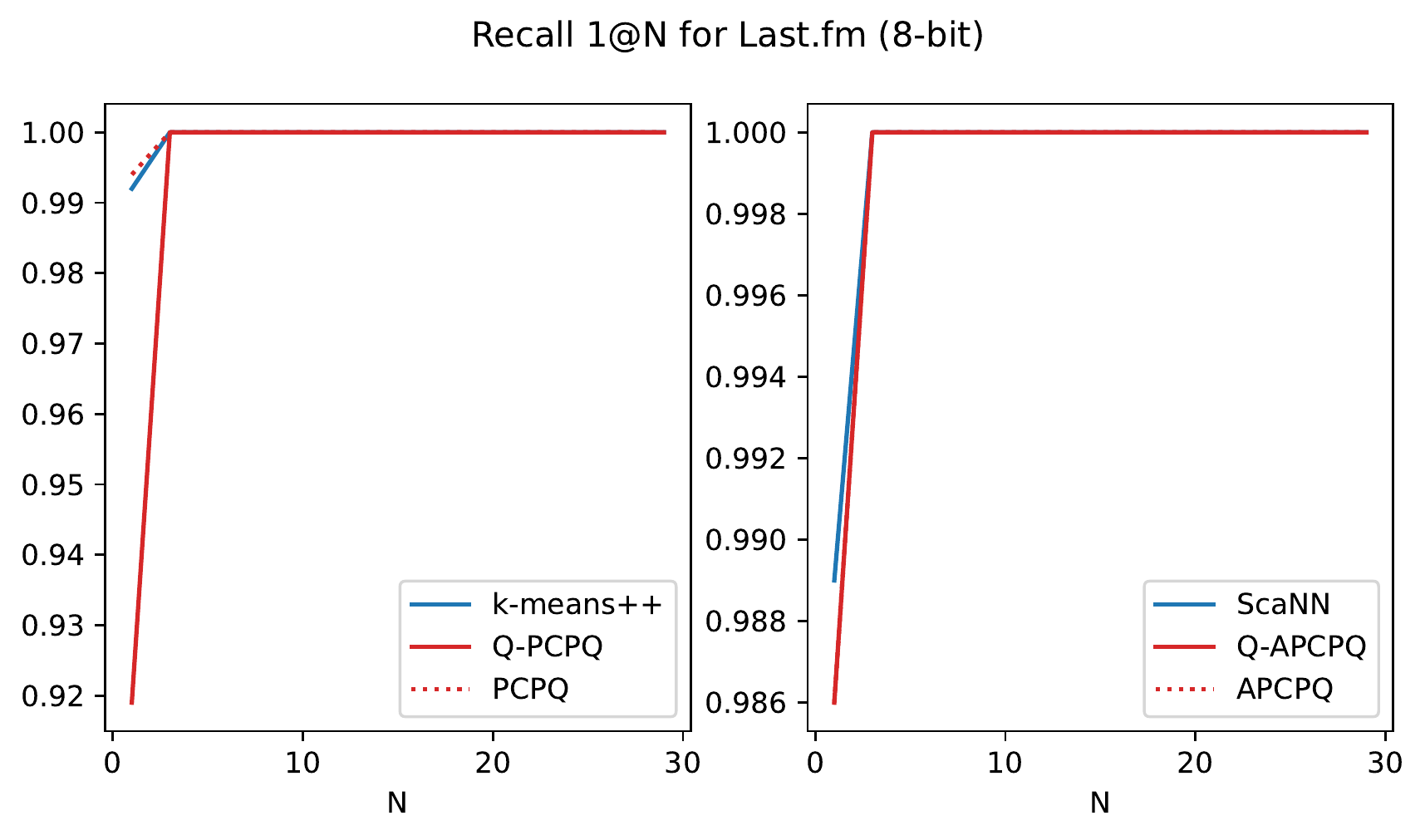}
  \label{fig:recall_lastfm_plot_8bit}
\end{subfigure}%
\vspace{-1em}
\caption{Comparison of $k$-means++ and ScaNN with Q-PCPQ and Q-APCPQ (and without projective quantization: PCPQ and APCPQ) respectively on Recall1@N on Glove-100 and Last.fm datasets in the 4-bit and 8-bit settings.}
\label{fig:recall_plots}
\vspace{0.5em}
\hrule 
\vspace{-0em}
\end{figure}
In this section we compare between $k$-clustering and its projective counterparts we proposed on their performance for maximum inner-product search. Specifically, we consider $k$-means++ and anisotropic $k$-clustering (ScaNN) \cite{guo2020accelerating} and compare them to quantized projective $k$-clustering (Q-PCPQ) and quantized anisotropic projective $k$-clustering (Q-APCPQ). In addition, we add the comaprisons to the versions of projective clustering without quantizing the scalars, which we denote by PCPQ and APCPQ respectively. In fixed-bit-rate settings we analyze the retrieval performance (i.e. recall) and the quality of approximation of maximum inner-product values to show that Q-PCPQ and Q-APCPQ achieve significantly better performance on both measures in standard benchmark datasets for MIPS. 

\paragraph{Datasets.} We use Glove-100 introduced in \citet{pennington2014glove} which is a collection of 1.2 million word embeddings in 100-dimensions. Glove-100 is used in \citet{guo2020accelerating} and is also one of the main benchmarks in \citet{aumuller2017ann} for inner-product search; \citet{aumuller2017ann} is a widely-used tool for benchmarking approximate nearest neighbor algorithms. While Glove-100 is meant for cosine similarity search, we unit-normalize all the points at training time making MIPS equivalent to cosine similarlity. Additionally we use the Last.fm dataset, another benchmark dataset from \citet{aumuller2017ann} for MIPS, containing 300,000 points in 65-dimensions created by embedding soundtracks using a matrix factorization based recommendation model for music.

\paragraph{Multiscale quantization.}
Performing PQ on the entire dataset when the dataset has more than $10^6$ data points can lead to poor performance and can be prohibitive for datasets of size $10^7$ or $10^9$. In order to circumvent this, typically in similarity search systems \cite{johnson2019billion,wu2017multiscale,guo2020accelerating} an initial clustering of the dataset is done, followed by a product quantization of the data points in each cluster. 

Usually, like in the aforementioned results, the residuals of the data points in each cluster are used for the PQ instead of the data points themselves. Specifically, for data points $X \in \R^{n \times d}$ a parameter $\bar{k} \in \N$ is chosen and $\bar{k}$ cluster centers $c_1, \dots, c_{\bar{k}} \in \R^d$ are computed. Typically $\bar{k}$ is chosen to be around $O(\sqrt{n})$ in order to balance the accucracy-performance tradeoff\footnote{see \texttt{https://github.com/facebookresearch/faiss/wiki/Guidelines-to-choose-an-index}}. Each point in $X$ is then mapped to its closest cluster center (in $\ell_2$-distance) to obtain a clustering $\{\Cc_1, \dots, \Cc_{\bar{k}}\}$ that partitions $X$. Finally, in order to build the index, the PQ method is applied to each cluster $\Cc_1, \dots, \Cc_{\bar{k}}$ separately. 

For queries, a parameter $k_{\text{probe}} \ll \bar{k}$ parameter is picked beforehand. On the arrival of a query $q \in \R^d$, points in the top $k_{\text{probe}}$ clusters are queried based on the ordering of $\iprod{q}{c_1}, \dots, \iprod{q}{c_{\bar{k}}}$ to obtain the point(s) with maximum inner-product with $q$.

\paragraph{Parameters.} For each dataset, we compute the initial clustering of $\bar{k}$ clusters using the \texttt{scikit-learn} implementation of $k$-means++ and set $\bar{k}$ so that $n/\bar{k} \approx 1000$. This provides a fair comparison for the retrieval performance and inner-product approximation of the PQ methods across datasets. For each dataset, the same clustering is used across all compared methods and parameters. 

We investigate two regimes for $k$, i.e., the number of centers in each section: $k = 16$ and $k = 256$ which are denoted the \emph{$4$-bit} and \emph{$8$-bit} settings respectively. We set the number of sections $m$ to be $m=25$ for Glove-100 and $m=16$ for Last.fm, this leads to $\bar{d} = 4$ coordinates per section. In the projected quantization setting (Q-PCPQ and Q-APCPQ), we set the scalar quantization to be $s = 8$.

For ScaNN and Q-APCPQ (and APCPQ), we set the threshold in the weight function to be $0.2 \times$ (average $\ell_2$-norm of the input vectors). An in-depth discussion of the rationale behind setting this threshold is done in \citet{guo2020accelerating}.  

\paragraph{Setup.} Our implementations for the methods are in Python v3.x and all our experiments were run on \texttt{n1-standard} Google Cloud Engine VMs with 4vCPUs @2.2GHz and 15GB RAM.

\subsection{Approximating maximum inner-product}
We compare the methods on their accuracy of estimating the top-1 inner-product. Specifically, we measure the relative error ${|\iprod{q}{x_i} - {\sum_{j = 1}^m \eta_{j, \phi_j(i)} \cdot \lambda_{\gamma(i)}}|}/{|\iprod{q}{x_i}|}$ for the point $x_i$ with the highest inner-product with $q$ over 1000 such queries. Here the summation $\sum_{j = 1}^m \eta_{j, \phi_j(i)} \cdot \lambda_{\gamma(i)}$ is the approximate inner-product for query $q$ from \eqref{eqn:final_ip_sum}. We perform a head-to-head comparison of each method in the 4-bit setting: i) $k$-means++ against Q-PCPQ (and PCPQ), and ii) ScaNN against Q-APCPQ (and APCPQ). The results over the 1000 queries are shown in Figure \ref{fig:ip_error_glove_plot} and the average relative error is given in Table \ref{fig:ip_error_table}. 

The results show that Q-PCPQ and Q-APCPQ consistently approximate the top inner-product with  the query better than $k$-means++ and ScaNN respectively. Specifically, the relative error is $12\%$ and $3\%$ lower for Q-PCPQ and Q-APCPQ respectively compared to their non-projective counterparts on Glove-100. Additionally, the results show that the loss in the approximation of the inner-product from quantizing the projections is negligible; as seen by the average relative errors between PCPQ and Q-PCPQ (and between APCPQ and Q-APCPQ) in Table \ref{fig:ip_error_table}. 

While Q-APCPQ does better on average than ScaNN, the average relative error is worse for Q-APCPQ than for Q-PCPQ. This does not align with the rationale of using the score-aware loss from \eqref{eqn:anisotropic_loss}, since the score-aware loss is meant to approximate top inner-products better than its counterpart in \eqref{eqn:pq_loss} that is not score-aware. We suspect that this might be due to the fact that we do not solve the minimization problem given in \eqref{eqn:proj-anisotropic-loss} for the anisotropic projective $k$-clustering optimally when $k=1$. 


\subsection{Recall}
\begin{table}[t]
\centering
\begin{tabular}{|l|ll|ll|ll|ll|}
\hline
                           & \multicolumn{4}{c|}{Glove-100}                                    & \multicolumn{4}{c|}{Last.fm}                                      \\
\hline
\multirow{2}{*}{Method}    & \multicolumn{2}{l|}{Recall 1@1} & \multicolumn{2}{l|}{Recall 1@10} & \multicolumn{2}{l|}{Recall 1@1} & \multicolumn{2}{l|}{Recall 1@10} \\
                           & 4-bit          & 8-bit         & 4-bit          & 8-bit          & 4-bit          & 8-bit         & 4-bit          & 8-bit          \\
\hline
\hline
$k$-means++                  & 0.448          & 0.772         & 0.825          & 0.945          & 0.584          & 0.992         & 0.984          & 1.000          \\
Q-PCPQ                     & 0.633          & 0.829         & 0.924          & 0.947          & 0.775          & 0.919         & 1.000          & 1.000          \\
PCPQ                       & 0.634          & 0.840         & 0.925          & 0.947          & 0.757          & 0.994         & 1.000          & 1.000          \\
\hline
ScaNN                      & 0.416          & 0.741         & 0.820          & 0.944          & 0.594          & 0.989         & 0.978          & 1.000          \\
Q-APCPQ                    & 0.529          & 0.766         & 0.889          & 0.945          & 0.667          & 0.986         & 1.000          & 1.000          \\
APCPQ                      & 0.541          & 0.770         & 0.895          & 0.945          & 0.674          & 0.986         & 1.000          & 1.000         \\
\hline
\end{tabular}
\caption{Recall1@1 and Recall1@10 values for $k$-means++, ScaNN, Q-PCPQ (and PCPQ) and Q-APCPQ (and APCPQ) for Glove-100 and Last.fm (averaged over 1000 queries) in the 4-bit and 8-bit regimes.}
\label{table:recall}
\vspace{0.5em}
\hrule 
\vspace{-0em}
\end{table}
One of the most important measures of the performance of a MIPS system is its retrieval performance. Specifically, for a set of queries $Q$ and a parameter $N \in \N^+$, let $\pi_N(q)$ be the top-$N$ points based on the approximate inner-products computed by method under consideration. Then, the Recall1@$N$ for the set of queries $Q$ is the following quantity: 
$$\text{Recall1@$N$} = \frac{1}{|Q|}\sum_{q \in Q}\id\left[\max_{x \in \pi_N(q)} \iprod{x}{q} \geq \max_{x \in X} \iprod{x}{q}\right].$$
We compare the performance of each method in 4-bit and 8-bit settings and measure the Recall1@1 and Recall1@10 for each dataset in Table \ref{table:recall}. Figure \ref{fig:recall_plots} plots the Recall1@N for Glove-100 and Last.fm in the 4-bit and 8-bit setting. Especially in the 4-bit setting, we see that Q-PCPQ and Q-APCPQ achieve significant recall gains over their non-projective counterparts ($k$-means++ and ScaNN); achieving approximately $8\%-10\%$ recall gains in the anistropic setting for Recall1@1 and close to $~19\%$ gain for Q-PCPQ in the same setting!

In the 8-bit regime, the number of centers per section is comparable to the number of data points, i.e., $k=256$ for 1000 points. Hence, even vanilla $k$-means++ clustering fits the data very well, achieveing as little as $3.16\%$ reconstruction error on Glove-100. In this setting we have that $k \approx n$, precluding the need for product quantization in the first place. This could explain the fact that Q-PCPQ and Q-APCPQ have much more modest gains in the 8-bit regime, and even does worse on Recall1@1 on Last.fm in the 8-bit regime. On much larger datasets however, where the number of centers $k$ per section is much smaller than the number of points being indexed, we expect Q-PCPQ to do much better than $k$-means, as is consistently shown in the 4-bit regime, where $k$-means++ achieves $30.09\%$ reconstruction error on Glove-100.

\bibliography{refs}

\appendix
\section{Projective Clustering}
\begin{lemma}\label{appx:proj_clustering_reduction} When $\Qc$ is isotropic over $\R^{\bar{d}}$, we have that \begin{align*}
\min_{c_1, \dots c_k \in \R^{\bar{d}}} \sum_{i=1}^n \min_{\alpha \in \R, j \in [k]} \ExpSub{q \sim \Qc}{\iprod{q}{ x - \alpha \cdot c_j}^2} = \min_{c_1, \dots, c_k \in \R^{\bar{d}}} \sum_{i = 1}^n \min_{j \in [k]} \left\|x_i - \frac{\iprod{x_i}{c_j}}{\|c_j\|_2^2} \cdot c_j \right \|_2^2.
\end{align*} 
\end{lemma}
\begin{proof}
\begin{align*}
\min_{c_1, \dots c_k  } \sum_{i=1}^n \min_{\alpha, j} \ExpSub{q \sim \Qc}{\iprod{q}{ x - \alpha \cdot c_j}^2} &= \min_{\{c_j\}_{j\in [k]}} \sum_{i = 1}^n \min_{\alpha, j} (x - \alpha c_j)^\top \Exp{qq^\top}(x - \alpha c_j) \\ 
&= \min_{\{c_j\}_{j\in [k]}} \sum_{i = 1}^n \min_{\alpha, j} \|(x - \alpha c_j)\|_2^2 \\ 
&= \min_{\{c_j\}_{j\in [k]}} \sum_{i = 1}^n \min_{j} \left\|(x - \frac{\iprod{x}{c_j}}{\|c_j\|_2^2} c_j)\right\|_2^2.
\end{align*}
The second-to-last equality follows from the isotropy of $\Qc$ and the last equality follows from the fact that $\frac{\iprod{x}{c_j}}{c_j}c_j$ is the projection of $x$ onto $c_j$.
\end{proof}

\section{Quantized Projective $k$-Clustering}
\subsection{Proof of Fact \ref{fact:qpcpq_loss_reduction}}\label{appx:qpcpq}
\begin{align*}
\min_{\{\bar{u}_j\}_{j \in [k]} } \sum_{j = 1}^k \left\|X_j - \bar{u}_{j} v_j^\top \right \|_F^2 &= \min \sum_{j = 1}^k \left\|(X_j - \sigma_j {u}_{j} v_j^\top) + (\sigma_j {u}_{j} v_j^\top - \bar{u}_{j} v_j^\top)  \right \|_F^2  \\ 
&\leq \min \sum_{j = 1}^k \|(X_j - \sigma_j {u}_{j} v_j^\top)\|_F^2 + \|(\sigma_j {u}_{j} v_j^\top - \bar{u}_{j} v_j^\top)  \|_F^2 \\ 
&= OPT + \min \sum_{j = 1}^k \|(\sigma_j {u}_{j} v_j^\top - \bar{u}_{j} v_j^\top) \|_F^2 \\ 
&= OPT + \min \sum_{j = 1}^k \|((\sigma_ju_j - \bar{u}_j) v_j^\top  \|_F^2 = OPT + \min \sum_{j = 1}^k \|\sigma_ju_j - \bar{u}_j\|_2^2 
\end{align*}
where the last line follows from the fact that $v_j^\top$ is a one-dimensional rotation matrix.

\section{Anisotropic Loss}\label{appx:hvalues_discussion}
\begin{wrapfigure}{r}{0.41\linewidth}
	\vspace{-1em}
	\centering
	\includegraphics[width=.35\textwidth]{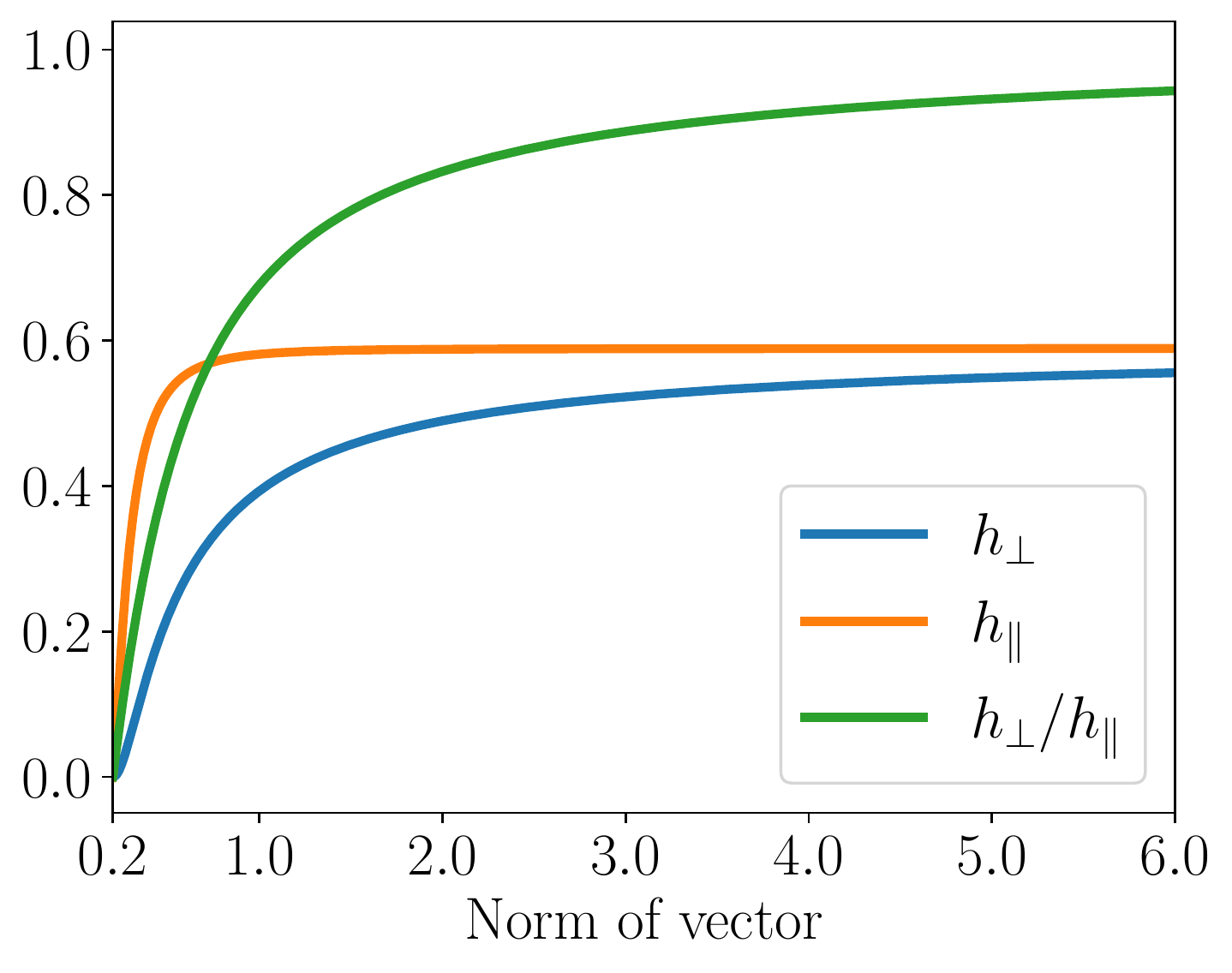}
	\vspace{-1em}
	\caption{$h_\parallel, h_\bot$ and $h_\bot/h_\parallel$ for $t = 0.2$ and $d = 4$.}
	\label{fig:hvalues}
	\vspace{0.5em}
	\hrule
	\vspace{-1.5em}
\end{wrapfigure}
When all $x \in X$ have identical $\ell_2$-norm, it is sufficient to consider the ratio $h_\bot/h_\parallel$. We plot the value of $h_\bot/h_\parallel$ for $t = 0.2$ in Figure \ref{fig:hvalues} to illustrate how the loss function changes in terms of the norm of the input vector $x$. 

As illustrated in the figure, as $\|x\|_2$ approaches $t$ (from above) the ratio $h_\bot/h_\parallel$ tends to $0$ indicating that only the parallel component of the residual matters, i.e. just the projection of the center onto $x$ is penalized. Although, when $\|x\|_2 \gg t$, the ratio $h_\bot/h_\parallel$ approaches $1$ resulting in weighting the parallel and residual components equally. This reduces the loss function to the loss function for $k$-means clustering: $\min_c \sum_{i = 1}^n h(\|r_\parallel(x_i, c)\|_2^2 + \|r_\bot(x_i, c)\|_2^2) = \min_c \sum_{i = 1}^n \|x_i - c\|_2^2.$


\end{document}